%% file: main.tex
\documentclass[11pt]{article}
\usepackage{fullpage}
\usepackage{subcaption}
\usepackage[usenames,dvipsnames]{xcolor}
\usepackage[colorlinks,citecolor=blue,linkcolor=BrickRed]{hyperref}
\usepackage{times}
\usepackage{latexsym,graphicx,epsfig,color}
\usepackage{amsfonts,amssymb,amsmath,amsthm,amstext}
\usepackage{enumitem}
\usepackage{url,setspace}
\usepackage{multirow}
\usepackage{rotating}
\usepackage{makeidx}
\usepackage{tikz}
\usetikzlibrary{calc}
\usepackage{enumitem}
\usepackage{accents}
\usepackage{xspace}
\usepackage{algorithm}
\usepackage[noend]{algpseudocode}
\usepackage{bm}
\usepackage{bbm}
\usepackage{cleveref}

\newcommand{\IGNORE}[1]{}

\usetikzlibrary{decorations.markings}
\usetikzlibrary{arrows}
\tikzstyle{block}=[draw opacity=0.7,line width=1.4cm]
\tikzstyle{graphnode}=[circle, draw, fill=black!20, inner sep=0pt, minimum width=6pt]
\tikzstyle{point}=[circle, draw, fill=black!30, inner sep=0pt, minimum width=1pt]
\tikzstyle{input}=[rectangle, draw, fill=black!75,inner sep=3pt, inner ysep=3pt, minimum width=4pt]
\tikzstyle{unmatched}=[graphnode,fill=black!0]
\tikzstyle{shaded}=[graphnode,fill=black!20]
\tikzstyle{matched}=[graphnode,fill=black!100]  	
\tikzstyle{matching} = [ultra thick]
\tikzset{
	>=stealth',
	pil/.style={
		->,
		thick,
		shorten <=2pt,
		shorten >=2pt,}
}
\tikzset{->-/.style={decoration={
			markings,
			mark=at position .5 with {\arrow{>}}},postaction={decorate}}}

\makeindex

\newtheorem{theorem}{Theorem}[section]

\newtheorem{lemma}[theorem]{Lemma}

\theoremstyle{definition}

\newtheorem{defn}[theorem]{Definition}

\newcounter{note}[section]

\newcommand{\OPT}{\textsc{OPT}}

\newcommand{\x}{\textbf{x}}

\def\eps {\epsilon}

\DeclareMathOperator*{\argmin}{arg\,min}
\newcommand{\ckc}{\ensuremath{\text{cost}}\xspace}
\newcommand{\vn}{\ensuremath{\Gamma\xspace}}

\usepackage{authblk}

\title{More than one Author with different Affiliations}
\author[1]{D Ellis Hershkowitz\thanks{dhershko@cs.cmu.edu}\thanks{Supported in part by NSF grants CCF-1527110, CCF-1618280, CCF-1814603, 
		CCF-1910588, NSF CAREER award CCF-1750808 and a Sloan Research Fellowship.}}
\author[2]{Gregory Kehne\thanks{gkehne@andrew.cmu.edu}}
\affil[1]{Computer Science Department, Carnegie Mellon University}
\affil[2]{Department of Mathematical Sciences, Carnegie Mellon University}

\begin{document}

\title{Reverse Greedy is Bad for $k$-Center}

\date{}
\maketitle

\begin{abstract}
We show the reverse greedy algorithm is between a $(2k-2)$- and a $2k$-approximation for $k$-center.	
\end{abstract}

\section{Introduction}
Whereas greedy algorithms \emph{greedily add to} a solution provided feasibility is maintained, reverse greedy algorithms \emph{greedily remove from} a solution until feasibility is achieved. For instance, in the $k$-median problem we have to choose $k$ facilities from among $n$ points such that the total distance of all points to their nearest chosen facility is minimized. The reverse greedy algorithm for $k$-median begins with all points as facilities and then repeatedly removes the facility whose removal increases the $k$-median cost the least until only $k$ facilities remain. 

Reverse greedy algorithms have proven useful for $k$-median where greedy algorithms have not. For example, in a surprising result, Chrobak et al. \cite{chrobak2005reverse} showed that while the greedy algorithm for $k$-median can perform arbitrarily badly, reverse greedy is an $O(\log n)$-approximation.\footnote{This guarantee is surpassed by $O(1)$-approximations based on LP rounding; see, for example, Charikar et al. \cite{charikar1999constant} for the first $O(1)$-approximation. Also see Byrka et al. \cite{byrka2014improved} for the best current approximation algorithm for $k$-median is a $(2.675 + \eps)$- approximation.} Moreover, the fact that reverse greedy removes from rather than adds to a solution allowed Anthony et al. \cite{anthony2008plant} to combine reverse greedy with multiplicative weights to solve a two-stage version of $k$-median. However, the only facility location problem for which reverse greedy has been studied is $k$-median.

We study reverse greedy for $k$-center.
In the $k$-center problem we must choose $k$ centers in a metric so that the maximum distance from any point in the metric to its nearest center is minimized.
$k$-center is hard to approximate to within better than a factor of $2$ \cite{hsu1979,hochbaum1984}, and the natural greedy algorithm is a $2$-approximation \cite{hochbaum1985best}.
Since the reverse greedy algorithm outperforms the greedy algorithm for $k$-median and the greedy algorithm for $k$-center achieves the best possible approximation for a polynomial-time algorithm, one might naturally expect that the reverse greedy algorithm for $k$-center also achieves a good approximation.
If reverse greedy for $k$-center gave reasonable approximation guarantees, then it would provide a greedy heuristic which performs well for both $k$-center and $k$-median.
It might, then, inform solutions for problems which interpolate between $k$-center and $k$-median such as ordered $k$-median, which has been the subject of an exciting recent string of results \cite{byrka2018constant,chakrabarty2018interpolating,chakrabarty2019approximation}; most recently a $(5 + \eps)$-approximation for ordered $k$-median by Chakrabarty and Swamy \cite{chakrabarty2019approximation}.
We show that, surprisingly, reverse greedy attains no such good approximation. In particular we demonstrate that the natural reverse greedy algorithm attains an approximation factor of $\Theta(k)$ for $k$-center.

\begin{theorem}
Reverse greedy is between a $(2k-2)$-approximation and a $2k$-approximation for $k$-center.	
\end{theorem}

We begin with an upper bound in \Cref{lem:upper}. The principal novelty of our upper bound proof is a new potential function which enables a simple proof of reverse greedy's approximation ratio. Roughly speaking, our potential function describes the extent to which facilities in the optimal solution have been consolidated together in the metric space. We hope that our potential might prove useful in the analyses of other center problems. 

Despite the weakness of our $2k$ upper bound, we show in \Cref{lem:lower} that our upper bound is essentially tight, as witnessed by a $2k-2$ lower bound. Our lower bound construction exploits the fact that by greedily removing the facility which increases cost the least, reverse greedy can repeatedly remove peripheral facilities until the final $k$ facilities lie in a single tightly packed region in the metric. Our lower bound holds for any $n = \Omega(k^2)$; thus, roughly speaking, reverse greedy is an $O(\log n)$-approximation on k-median but a much worse $\Omega(\sqrt{n})$-approximation on k-center.

\section{$k$-Center and Reverse Greedy}
We now more formally describe the $k$-center problem and the reverse greedy algorithm. In the $k$-center problem an algorithm must pick $k$ centers such that the maximum distance of every point (a.k.a.\ client) to a chosen facility is minimized. Formally, an instance of $k$-center is given by a metric space $(d, \mathcal{C})$  and an integer $k \in \mathbb{N}$, where $d(c_1, c_2) \in \mathbb{R}^+$ gives the distance between $c_1, c_2 \in \mathcal{C}$. An algorithm must output a set of facilities $F \subseteq \mathcal{C}$ where $|F| \leq k$. Note that we make the standard assumption that every possible facility is also a client. The cost of a solution $F$ is $\ckc(F) :=\max_{c \in \mathcal{C}} d(c, F)$, where $d(c, F) := \min_{f \in F} d(c, f)$ for $c \in \mathcal{C}$ and $F \subseteq \mathcal{C}$. If $f' \in F$ is the unique facility for which $d(c,f') = \min_{f \in F} d(c, f)$ then we say that $f'$ \emph{serves} $c$.

The reverse greedy algorithm for $k$-center repeatedly removes the facility that increases the $k$-center cost the least, breaking ties arbitrarily, until only $k$ facilities remain:

\begin{algorithm}
	\renewcommand{\thealgorithm}{} 
	\floatname{algorithm}{}
	\caption{Reverse Greedy for $k$-Center}
\begin{algorithmic}	
	\State $F_0 \gets \mathcal{C}$
	\For{$i \in [n-k]$}
		\State $f_i \gets f \in \argmin_{f \in F_{i-1}} \ckc(F_{i-1} \setminus \{f\})$
		\State $F_i \gets F_{i-1} \setminus \{ f_i \}$
	\EndFor
	\Return $F_{n-k}$
\end{algorithmic}
\end{algorithm}


We let $\OPT := \ckc(F^*)$, where $F^* = \{o_1, \ldots, o_k\}$ is an optimal solution to $k$-center, and let $O_1, \ldots, O_k$ denote the corresponding $F^*$ balls, where $O_t = \{c \in \mathcal{C} : d(o_t, c) \leq \OPT \}$.

\section{Upper Bound}
We now prove that reverse greedy is a $2k$-approximation. To begin, as a simple observation notice that if $| F_{n-k} \cap O_t | \geq 1$ for every $t$ then by the triangle inequality reverse greedy is a $2$-approximation and so clearly a $2k$-approximation for any $k \geq 1$. Thus, for the remainder of this section we will be focused on the case where $|F_{n-k} \cap O_t | = 0$ for some $t$ which by the pigeonhole principle implies that $|F_{n-k} \cap O_{t'}| \geq 2$ for some $t'$.

Roughly, our proof strategy will be to show that after a cost increase of \emph{about} $2 \cdot \OPT$ we make 1 unit of progress for a potential function that starts at $k$ and is always positive. In order to execute this strategy we identify a sequence of reverse greedy states each of which costs about ${2 \cdot \OPT}$ more than the previous. Specifically, we consider all $l$-critical iterations of reverse greedy, defined as follows:
\begin{defn}[$l$-Critical, $c_l$]
	$F_i$ is $l$-critical if $\ckc(F_i) \leq 2l \cdot \OPT$ and $\ckc(F_{i+1}) > 2l \cdot \OPT$.	Let $c_l$ be the step $i$ for which $F_i$ is $l$-critical. We say $c_l$ is defined if such an $i$ exists.
\end{defn}

We will consider the critical states $F_0 = F_{c_0}, F_{c_1}, F_{c_2}, \ldots$ of reverse greedy and present a suitable potential which starts with value $k$ at state $F_{c_0}$, is always strictly greater than $0$, and decreases by at least $1$ between each pair of $F_{c_l}$ and $F_{c_{l+1}}$.
A natural candidate for our potential function is the number of $O_t$s with non-empty intersection with reverse greedy's solution. 
However, it is not too hard to see that reverse greedy does not necessarily always empty an $O_t$ from $F_{c_l}$ to $F_{c_{l+1}}$. 
Instead, we will measure our progress more cleverly.
In particular, we will show that between $F_{c_l}$ and $F_{c_{l+1}}$ either reverse greedy removes all facilities from some $O_t$ or it consolidates two sets of $O_t$s.
Roughly speaking, collections of $O_t$s are consolidated when their only remaining facilities are those in their intersection. 

Formally, we define the extent to which the $O_t$s have been consolidated with the following two notions:
\begin{defn}[Consolidation]
	Given an instance of $k$-center with optimal solution balls $O_1, \ldots, O_{k}$ defined as above, we define a consolidation of facilities $F \subseteq \mathcal{C}$ to be a collection $\Phi = \{ P_1, P_2, \ldots\}$ of $P_s \subseteq \mathcal{C}$ with the following properties:
	\begin{enumerate}
		\item $F \subseteq \bigcup_{s} P_s$; \qquad (Covering) 
		\item $\max_{x, y \in P_s}d(x, y) \leq 2 \cdot \OPT$ for all $P_s$; \qquad (Diameter)
		\item If $f, f' \in F \cap O_t$ then $f,f' \in P_s$ for some $P_s$. \qquad (Optimal Pairs)
	\end{enumerate}
\end{defn}

\begin{defn}[Consolidation Number]
	Let the consolidation number $\vn(F)$ of facilities $F\subseteq \mathcal{C}$ be the minimum cardinality of a consolidation of $F$.
\end{defn}

We begin by observing that $1 \leq \vn(F) \leq k$ for all nonempty $F \subseteq \mathcal{C}$, since a consolidation must cover $F$ and since ${\Phi = \{O_1, \ldots, O_k\}}$ is a consolidation of any $F\subseteq \mathcal{C}$. It is easy to verify that, more generally, if $\Phi$ is a consolidation of $F$ then it is also a consolidation of any subset of $F$.

Our next lemma shows that the consolidation number decreases by at least $1$ from each critical state to the next. For $f \in \mathcal{C}$ we define ${B_r(f) := \{f' \in \mathcal{C} : d(f, f') \leq r \cdot \OPT \}}$ to be the closed ball of radius $r \cdot \OPT$ centered at $f$.

\input{upper_bound_fig}
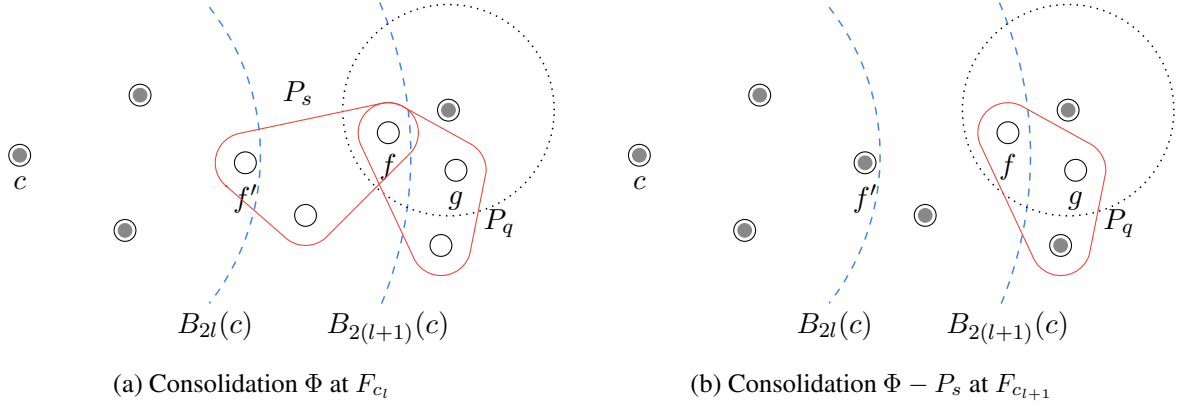
\begin{figure}
	\centering
	\begin{subfigure}[t]{0.4\textwidth}
		\centering
		\begin{tikzpicture}[scale=0.8*\tikzScale]
		\drawFirstUBFig
		\end{tikzpicture}
		\caption{Consolidation $\Phi$ at $F_{c_l}$} \label{fig:upperBound1} 
	\end{subfigure}
	\qquad \qquad
	\begin{subfigure}[t]{0.4\textwidth}
		\centering
		\begin{tikzpicture}[scale=0.8*\tikzScale]
		\drawSecondUBFig
		\end{tikzpicture}
		\caption{Consolidation $\Phi - P_s$  at $F_{c_{l+1}}$} \label{fig:upperBound2} 
	\end{subfigure}
	\caption{The configuration of clients, facilities, balls, and consolidation sets described in the proof of \Cref{lem:VNDec}. Points in $\mathcal{C}$ are small circles; points removed by reverse greedy are colored in grey; $O_t$ is larger dotted circle; edges of $B_{2l}(c)$ and $B_{2(l+1)}(c)$ are dashed in blue; $P_s$ and $P_q$ are outlined in red. Facilities $f'$ and $f$ serve client $c$ for $F_{c_l}$ and $F_{c_{l+1}}$, respectively. Set $P_s$ may be removed from the consolidation at $F_{c_{l+1}}$ because $P_q$ also covers $f$.} \label{fig:upperBound} 
\end{figure}

\begin{lemma}[$\vn$ decreases by $1$]\label{lem:VNDec}
	If there exists an $O_t$ such that $|O_t \cap F_{n-k}| \geq 2$ then $\vn(F_{c_{l+1}}) < \vn (F_{c_l})$ for all $l$ such that $c_l$ and $c_{l+1}$ are defined.
\end{lemma}

\begin{proof}	
	Let $f$ and $g$ be a pair of facilities in $O_t \cap F_{n-k}$.
	Because $f$ is present in $F_{c_{l+1}}$ and $F_{c_{l+1}}$ is $(l + 1)$-critical, in $F_{c_{l+1}}$ we know that $f$ is serving some client $c$ for which the next nearest facility $\bar{f} \in F_{c_{l+1}}$ has cost ${d(c, \bar{f}) > 2(l+1)\cdot\OPT}$.
	In other words, ${B_{2(l+1)}(c) \cap F_{c_{l+1}} = \{f\}}$.
	This $c$ must have been served by some facility $f' \in F_{c_l}$ for which ${d(c, f')\leq 2l\cdot \OPT}$. 
	Let $\Phi$ be a minimal consolidation for $F_{c_l}$, and choose a $P_s \in \Phi$ for which $f' \in P_s$. Such a $P_s$ must exist by the covering property of consolidations. 
	See \Cref{fig:upperBound1}.
	We will argue that $\Phi - P_s$ is a consolidation for $F_{c_{l+1}}$, and so
	\[
	\vn(F_{c_{l+1}}) \leq |\Phi - P_s| < |\Phi| = \vn(F_{c_l}).
	\]
	
	To see this, first observe that by the diameter property of consolidations, $P_s \subseteq B_2(f')$, and so since ${d(c, f')\leq 2l\cdot \OPT}$, the triangle inequality implies that ${B_2(f') \subseteq B_{2(l+1)}(c)}$.	
	Therefore ${P_s \subseteq B_{2(l+1)}(c)}$, and so 
	\[
	P_s \cap F_{c_{l+1}} \subseteq B_{2(l+1)}(c) \cap F_{c_{l+1}} = \{f\}.
	\]
	If $P_s\cap F_{c_{l+1}} =\emptyset$, then it is easy to verify that $\Phi - P_s$ is a consolidation of $F_{c_{l+1}}$ and so $\vn(F_{c_{l+1}}) < \vn(F_{c_{l}})$.
	
	Otherwise $P_s \cap F_{c_{l+1}} = \{f\}$, as is the case in \Cref{fig:upperBound2}.
	Since $\Phi$ is a consolidation of $F_{c_{l}}$ and $\{f\}$ is a singleton, $\Phi - P_s$ still satisfies the diameter and optimal pair consolidation criteria. We need only worry, then, that $\Phi - P_s$ will fail to be a consolidation by ceasing to cover $F_{c_{l+1}}$---in particular by failing to cover $f$. 
	However this is not the case; there must be a $ P_q \in \Phi$ such that $f \in P_q \neq P_s$, as in \Cref{fig:upperBound2}.
	The existence of such a $P_q$ follows from the assumption that $f,g \in O_t$; then because $\Phi$ satisfies the optimal pair consolidation property and ${f, g \in F_{c_l}}$, there is some $P_q\in \Phi$ such that $f,g \in P_q$.
	Finally, $P_q \neq P_s$ because $P_s \subseteq B_{2(l+1)}(c)$ and ${d(c, g) > 2(l+1)\cdot \OPT}$, and so $g \not\in P_s$.
	Therefore there is a $P_q\in \Phi$ distinct from $P_s$ which covers $f$, and so $\Phi - P_s$ is indeed a consolidation of $F_{c_{l+1}}$, implying that $\vn(F_{c_{l+1}}) < \vn(F_{c_{l}})$.
\end{proof}

We now conclude our upper bound.

\begin{lemma} \label{lem:upper}
	Reverse greedy is a $2k$-approximation for $k$-center.
\end{lemma}

\begin{proof}
	If for every $O_t$ we have $|O_t \cap F_{n-k}| \geq 1$, then reverse greedy is clearly a $2$-approximation by the triangle inequality and therefore a $2k$-approximation for all $k \geq 1$. 
	Thus we may assume that there is some $O_t$ such that $|O_t \cap F_{n-k}| = 0$, and so by the pigeonhole principle there exists some $O_{t'}$ for which $|O_{t'} \cap F_{n-k}| \geq 2$.
	
	Let $\bar{l}$ be the largest $l$ such that $c_{l}$ is defined. By \Cref{lem:VNDec} we have $1 \leq \vn(F_{c_l}) - \vn(F_{c_{l+1}})$; applying this $\bar{l}$ times and recalling that $1\leq \vn(F) \leq k$, we find that
	\[
	\bar{l} \leq \vn(F_{c_0}) - \vn(F_{c_{\bar{l}} } ) \leq k-1.
	\]
	The maximality of $\bar{l}$ and the definition of $l$-criticality then imply that
	\begin{align*}
	\ckc(F_{n-k}) &< \ckc(F_{c_{\bar{l}}}) + 2 \cdot \OPT \\ 
	&\leq 2  \bar{l} \cdot \OPT + 2 \cdot \OPT \\
	& \leq 2 k \cdot \OPT,
	\end{align*}
	and therefore reverse greedy is a $2k$-approximation for $k$-center.
\end{proof}

\section{Lower Bound}
In this section we show that reverse greedy is at best a $(2k-2)$-approximation for $k$-center. 

We aim to formalize the intuition that by greedily removing the facility which increases cost the least, reverse greedy can repeatedly remove peripheral facilities until the final $k$ facilities lie in a single tightly packed region in the metric. If reverse greedy ends with its facilities packed into single region, then it must have been the case that for each facility $f \in F_{n-k}$ and each iteration $i$, $f$ served a client that had no alternative facility within distance $\ckc(F_{i})$. Thus, to produce an instance of $k$-center where reverse greedy performs badly we must produce a metric with a tightly packed region of centers where in every iteration of reverse greedy, each one of these centers serves some client whose second closest center is further than $\ckc(F_{i})$.

\input{lower_bound_fig}
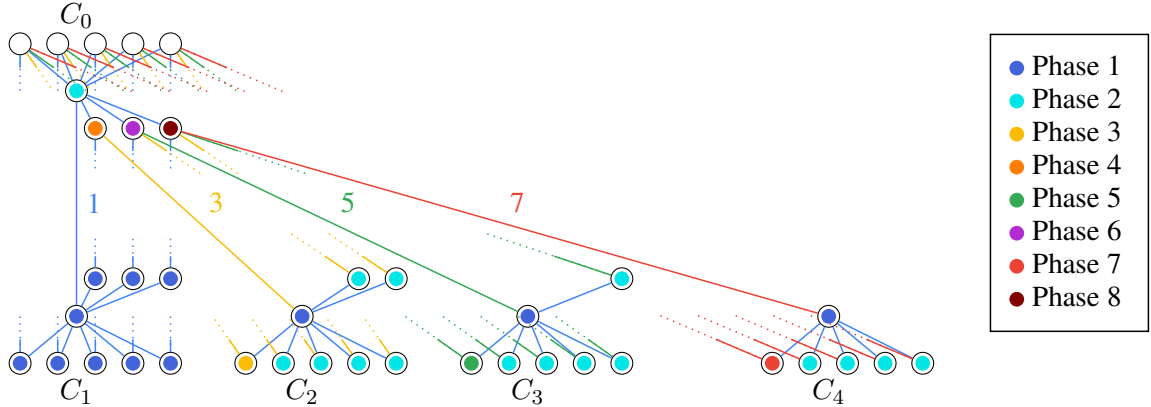
\begin{figure}
	\centering
	\begin{tikzpicture}[scale=\tikzScale]
	\drawBigFig
	\drawBigFigKey
	\end{tikzpicture}
	\caption{An execution of reverse greedy on the metric completion of $G$ for $k=5$ and $n=39$. Facilities are color-coded according to total cost at the time of their removal. The optimal solution gives $O_i = C_i$, while this execution chooses the five unmarked facilities in $C_0$.} \label{fig:lowerBound} 
\end{figure}

Formally, we construct an instance of $k$-center given by $(\hat{d}, \hat{\mathcal{C}})$ for a given $k$ and ${n = \frac{1}{2}(3k-2)(k+1)}$ as follows.
Our $\hat{d}$ will be the shortest path metric of a weighted graph $G$, and so we begin by describing $G$.
First, take the sequence of stars of descending size ${S_{2k-1}, S_{2k-1}, S_{2k-2}, S_{2k-3}, \ldots, S_{k+1}}$, which we will refer to as $C_0, C_1, C_2, \ldots, C_{k-1}$.
Notice that the first two stars ($C_0$ and $C_1$) have the same size.
We will assume that the leaves of each star have some canonical ordering where \emph{corresponding} leaves of $C_i$ and $C_j$ are leaves given the same number in $C_i$ and $C_j$ respectively.
The edges of the stars have weight 1.
Next match the centers of $C_0$ and $C_1$ with weight 1 edges and match the leaves of $C_0$ and $C_1$ with weight 1 edges.
For each subsequent $C_i$, match the leaves of $C_i$ with the corresponding leaves of $C_0$ and match the center of $C_i$ to the $(2k-i)$th leaf of $C_0$, all with edges of weight $2i - 1$.
We illustrate $G$ in \Cref{fig:lowerBound}.

We derive $(\hat{d}, \hat{\mathcal{C}})$ by taking $\hat{\mathcal{C}} = V(G)$ and $\hat{d}$ to be the metric completion of $G$.
Observe that $\OPT =1$ in this instance, since there are $k$ stars $C_i$, and so a feasible solution is to choose the center of each $C_i$.
We now argue that reverse greedy performs poorly on $(\hat{d}, \hat{\mathcal{C}})$.

\begin{lemma}\label{lem:lower}
For every $k$ and every ${n \geq \frac{1}{2}(3k-2)(k+1)}$, there exists an instance of $k$-center for which reverse greedy returns a solution of cost $(2k-2)\cdot \OPT$.
\end{lemma}

\begin{proof}
Consider $(\hat{d}, \hat{\mathcal{C}})$ as described above and illustrated in \Cref{fig:lowerBound}.
We provide a particular series of choices that reverse greedy could make on $(\hat{d}, \hat{\mathcal{C}})$ for any $k$ and ${n =\frac{1}{2}(3k-2)(k+1)}$.
It is easy to see that our analysis can be adapted to the case when ${n > \frac{1}{2}(3k-2)(k+1)}$ by adding additional leaves to $C_0$ and then having reverse greedy immediately remove these leaves from its solution.

We will split our analysis of these choices into `phases' of facility removals, where the Phase $r$ is the set of iterations for which reverse greedy's solution costs $r$, for $r \in [2k-2]$.
That is, the $r$th phase consists of the set of $i$ for which $\ckc(F_i) = r$. See  \Cref{fig:lowerBound} for an example of the phases of reverse greedy we consider in our analysis.

Notice that after removing all facilities in $C_1$ and the center facilities in $C_2, \ldots, C_{k-1}$, the cost of reverse greedy's solution is $1$, and no other facilities can be removed without increasing the cost.
Thus, let Phase $1$ be this sequence of removals.
Next, observe that the center of $C_0$ and all but one leaf of each of $C_2, \ldots, C_{k-1}$ may be removed while incurring a cost of 2.
The removal of the center of $C_0$ requires that the center of $C_1$ be served by another facility in $C_0$ at cost 2, and the removal of the other leaves of $C_2, \ldots, C_{k-1}$ require that they be served at cost 2 by the remaining leaves.
Thus, we let Phase 2 consist of the removal of the center of $C_0$ along with the removal of all but one leaf from every other $C_i$. 

We now argue inductively that given Phase $1$ and Phase $2$, over the course of the subsequent pairs of phases $r=2i-1$ and $r'=2i$ from $i =2$ to $i=k-1$, reverse greedy first empties $C_i$ and then removes a single facility from $C_0$.
Notice that once all facilities from $C_{i-1}$ have been removed, they `bind' the facilities in $C_0$: if any leaf node in $C_0$ is removed then the leaf node in $C_{i-1}$ to which it is matched would have to travel distance at least $2i-1$ to the next nearest facility.
Thus, we have that after Phase $(2i-2)$ removing any facility in $C_0$ would increase our cost to $r$.
Moreover, we know that removing the single facility from $C_i$ also increases our cost to $r$, since doing so would cause every leaf node in $C_i$ to travel $2i-1$ to its matched leaf node in $C_0$, and the center of $C_i$ to travel to the $(2k-i)$th leaf of $C_0$.
Additionally, removing any facility from $C_j$ for $j > i$ increases the cost to $> 2i$.
Thus, in Phase $r$ let reverse greedy remove the single remaining leaf from $C_i$ and leave $C_0$ untouched. 

Now all leaves in $C_0$ are serving clients in $C_i$ at cost $r$, and if any were removed then their corresponding leaves in $C_i$ would be served at cost $r+1=r'$.
Thus, for Phase $r'$, we let reverse greedy remove the leaf of $C_0$ to which the center of $C_i$ is matched. 
Then this center client is served at a cost of $r'$ (by a center in $C_0$ with which one of the leaves in $C_i$ is matched), and if any facility from $C_0$ were removed then its matched client in $C_i$ would have cost $r'+1$. At this point the stage is set for the next pair of phases.

This series of phases terminates with all $C_1, \ldots, C_{k-1}$ empty and $k$ leaves of $C_0$ remaining. The center of $C_{k-1}$ is served at a cost of $2k-2$, implying that the final cost of reverse greedy is at least $(2k-2) \cdot \OPT$. 
\end{proof}

\section{Future Work}
We conclude with two open questions regarding the performance of reverse greedy. 

Our first question deals with interpolation. As mentioned, previous work of Chrobak et al. \cite{chrobak2005reverse} shows that reverse greedy is an $O(\log n)$-approximation for $k$-median while our analysis has demonstrated that reverse greedy is a $\Omega(\sqrt{n})$-approximation for $k$-center. The former can be seen as optimization over an $l_1$ metric whereas the latter can be seen as optimization of an $l_\infty$ metric. A natural question, then, is what sort for approximation ratio reverse greedy attains for the corresponding $l_p$ problems for $1 < p < \infty$. For example, what approximation does reverse greedy attain on $k$-means, the corresponding $l_2$ optimization problem? It would be quite interesting if the approximation guarantee of reverse greedy smoothly and monotonically varied between these two extremes. 

Our second open question concerns the beyond-worst-case performance of reverse greedy for $k$-center. While our analyses have shown that \emph{in the worst case} reverse greedy performs poorly for $k$-center, it might very well be the case that reverse greedy attains a good approximation ratio under suitable beyond-worst-case assumptions. It is easy to verify that our lower bound construction is quite brittle, in that it fails to witness a $2k - 2$ lower bound under even small perturbations of its weights. Furthermore, one can show that if the $O_t$s of a solution are all at least $2 \cdot \OPT$ apart then reverse greedy always attains a $2$-approximation. Such a separation of the $O_t$s is similar to the beyond-worst-case notion of perturbation-resilience which has been studied in the center location literature \cite{awasthi2012center, balcan2012clustering}. This leads us to conclude with our second open question: is reverse greedy an $O(1)$-approximation under natural beyond-worst-case assumptions? Such a characterization could elevate reverse greedy from a lowly, worst-case bad algorithm to a simple, practical algorithm for $k$-center.

\section{Acknowledgments}
We would like to thank R. Ravi and Goran Zuzic for helpful discussions.

\bibliographystyle{alpha}
\bibliography{main}

\end{document}

%% file: upper_bound_fig.tex
\def\tikzScale{.5}
\def\chunkScale{1}

\tikzstyle{gregNode}=[circle,draw,fill=white, scale=0.8]
\tikzstyle{markNode}=[circle, scale=.5]
\tikzstyle{cliqueEdge}=[line width = .2mm,matchEdgeColorA]
\tikzstyle{matchEdge}=[line width = .2mm]
\tikzstyle{optBallStyle}=[circle, draw, minimum size =80, dotted, line width=.2mm]
\tikzstyle{ballStyle}=[line width = .2mm, matchEdgeColorA, dashed, line width=.2mm]

\definecolor{matchEdgeColorA}{HTML}{4285F4} 
\definecolor{matchEdgeColorB}{HTML}{FBBC05} 
\definecolor{matchEdgeColorC}{HTML}{34A853} 
\definecolor{matchEdgeColorD}{HTML}{EA4335} 
\definecolor{markColorn}{HTML}{888888}	
\definecolor{markColor0}{HTML}{FFFFFF}	
\definecolor{markColor1}{HTML}{4363D8}	
\definecolor{markColor2}{HTML}{00E6E6}	
\definecolor{markColor3}{HTML}{FBBC05}	
\definecolor{markColor4}{HTML}{FF8000}	
\definecolor{markColor5}{HTML}{34A853}	
\definecolor{markColor6}{HTML}{B12DD2}	
\definecolor{markColor7}{HTML}{EA4335}	
\definecolor{markColor8}{HTML}{800000}	

\newcommand{\convexPath}[2]{
	[   
	create hullnodes/.code={
		\global\edef\namelist{#1}
		\foreach [count=\counter] \nodename in \namelist {
			\global\edef\numberofnodes{\counter}
			\node at (\nodename) [draw=none,name=hullnode\counter] {};
		}
		\node at (hullnode\numberofnodes) [name=hullnode0,draw=none] {};
		\pgfmathtruncatemacro\lastnumber{\numberofnodes+1}
		\node at (hullnode1) [name=hullnode\lastnumber,draw=none] {};
	},
	create hullnodes
	]
	($(hullnode1)!#2!-90:(hullnode0)$)
	\foreach [
	evaluate=\currentnode as \previousnode using \currentnode-1,
	evaluate=\currentnode as \nextnode using \currentnode+1
	] \currentnode in {1,...,\numberofnodes} {
		-- ($(hullnode\currentnode)!#2!-90:(hullnode\previousnode)$)
		let \p1 = ($(hullnode\currentnode)!#2!-90:(hullnode\previousnode) - (hullnode\currentnode)$),
		\n1 = {atan2(\y1,\x1)},
		\p2 = ($(hullnode\currentnode)!#2!90:(hullnode\nextnode) - (hullnode\currentnode)$),
		\n2 = {atan2(\y2,\x2)},
		\n{delta} = {-Mod(\n1-\n2,360)}
		in 
		{arc [start angle=\n1, delta angle=\n{delta}, radius=#2]}
	}
	-- cycle
}

\def\mark(#1)[#2] { 
	\ifthenelse{#2 = -1}{
		\node[markNode, draw=markColorn, fill=markColorn] ($#1$ m) at (#1) {};}{}
	\ifthenelse{#2 = 0}{
		\node[markNode, draw=markColor0, fill=markColor0] ($#1$ m) at (#1) {};}{}
	\ifthenelse{#2 = 1}{
		\node[markNode, draw=markColor1, fill=markColor1] ($#1$ m) at (#1) {};}{}
	\ifthenelse{#2 = 2}{
		\node[markNode, draw=markColor2, fill=markColor2] ($#1$ m) at (#1) {};}{}
	\ifthenelse{#2 = 3}{
		\node[markNode, draw=markColor3, fill=markColor3] ($#1$ m) at (#1) {};}{}
	\ifthenelse{#2 = 4}{
		\node[markNode, draw=markColor4, fill=markColor4] ($#1$ m) at (#1) {};}{}
	\ifthenelse{#2 = 5}{
		\node[markNode, draw=markColor5, fill=markColor5] ($#1$ m) at (#1) {};}{}
	\ifthenelse{#2 = 6}{
		\node[markNode, draw=markColor6, fill=markColor6] ($#1$ m) at (#1) {};}{}
	\ifthenelse{#2 = 7}{
		\node[markNode, draw=markColor7, fill=markColor7] ($#1$ m) at (#1) {};}{}
	\ifthenelse{#2 = 8}{
		\node[markNode, draw=markColor8, fill=markColor8] ($#1$ m) at (#1) {};}{}
}

\def\drawFirstUBFig{
	\node[gregNode, label=below:$c$] at (0,0)   (c) {};		
	\node[gregNode] at (3.5,-2.5)   (chuff1) {};		
	\node[gregNode] at (4,2)   (chuff2) {};		
	\node[gregNode, label=below:$f'$] at (7.5,-0.25)   (fprime) {};		
	\node[gregNode] at (9.5,-2)   (chuff4) {};		
	\node[gregNode, label=below:$f$] at (12.25,0.75)   (f) {};		
	\node[gregNode] at (14.25,1.5)  (center) {};		
	\node[gregNode, label=below:$g$] at (14.5,-0.5)   (g) {};			
	\node[gregNode] at (14,-3)   (gprime) {};
	
	\draw (c) ++(38:8) arc (38:-38:8)[ballStyle];
	\node[auto, label=below:$B_{2l}(c)$] (l1) at (6.5,-4.5){};
	\draw (c) ++(23:13) arc (23:-23:13)[ballStyle];
	\node[auto, label=below:$B_{2(l+1)}(c)$] (l2) at (12.25,-4.5){};
		
	\node[optBallStyle] at (center) (O1) {};
	
	\draw[markColor7] \convexPath{fprime,f,chuff4}{1cm};
	\node[auto, color=markColor7, label=above:$P_s$] (p1) at (9.25,1){};

	\draw[markColor7] \convexPath{g,gprime,f}{1cm};
	\node[auto, color=markColor7, label=right:$P_q$] (p2) at (14.75,-2.25){};
			
	\mark(c)[-1]
	\mark(chuff1)[-1]
	\mark(chuff2)[-1]
	\mark(center)[-1]
}

\def\drawSecondUBFig{
	\node[gregNode, label=below:$c$] at (0,0)   (c) {};		
	\node[gregNode] at (3.5,-2.5)   (chuff1) {};		
	\node[gregNode] at (4,2)   (chuff2) {};		
	\node[gregNode, label=below:$f'$] at (7.5,-0.25)   (fprime) {};		
	\node[gregNode] at (9.5,-2)   (chuff4) {};		
	\node[gregNode, label=below:$f$] at (12.25,0.75)   (f) {};		
	\node[gregNode] at (14.25,1.5)  (center) {};		
	\node[gregNode, label=below:$g$] at (14.5,-0.5)   (g) {};			
	\node[gregNode] at (14,-3)   (gprime) {};
	
	\draw (c) ++(38:8) arc (38:-38:8)[ballStyle];
	\node[auto, label=below:$B_{2l}(c)$] (l1) at (6.5,-4.5){};
	\draw (c) ++(23:13) arc (23:-23:13)[ballStyle];
	\node[auto, label=below:$B_{2(l+1)}(c)$] (l2) at (12.25,-4.5){};
	
	\node[optBallStyle] at (center) (O1) {};
	
	
	\draw[markColor7] \convexPath{g,gprime,f}{1cm};
	\node[auto, color=markColor7, label=right:$P_q$] (p2) at (14.75,-2.25){};
	
	\mark(c)[-1]
	\mark(chuff1)[-1]
	\mark(chuff2)[-1]
	\mark(chuff4)[-1]
	\mark(center)[-1]
	\mark(fprime)[-1]
	\mark(gprime)[-1]
}

%% file: lower_bound_fig.tex
\def\tikzScale{.5}
\def\chunkScale{1}
\def\wholeScale{2}

\tikzstyle{gregNode}=[circle,draw,fill=white, scale=0.8]
\tikzstyle{markNode}=[circle, scale=.5]
\tikzstyle{cliqueEdge}=[line width = .2mm,matchEdgeColorA]
\tikzstyle{matchEdge}=[line width = .2mm]

\definecolor{matchEdgeColorA}{HTML}{4285F4} 
\definecolor{matchEdgeColorB}{HTML}{FBBC05} 
\definecolor{matchEdgeColorC}{HTML}{34A853} 
\definecolor{matchEdgeColorD}{HTML}{EA4335} 
\definecolor{markColor0}{HTML}{FFFFFF}	
\definecolor{markColor1}{HTML}{4363D8}	
\definecolor{markColor2}{HTML}{00E6E6}	
\definecolor{markColor3}{HTML}{FBBC05}	
\definecolor{markColor4}{HTML}{FF8000}	
\definecolor{markColor5}{HTML}{34A853}	
\definecolor{markColor6}{HTML}{B12DD2}	
\definecolor{markColor7}{HTML}{EA4335}	
\definecolor{markColor8}{HTML}{800000}	

\def\gregLine(#1,#2)[#3][#4][#5]{ 
	\ifthenelse{\equal{#5}{true}}{
		\draw #1 edge [#3,#4] ($#1!0.15!#2$) edge [dotted, #3, #4] ($#1!0.3!#2$);
		\draw #2 edge [#3,#4] ($#2!0.15!#1$) edge [dotted, #3, #4] ($#2!0.3!#1$);
	}{	
		\draw #1 edge [#3,#4] ($#1!0.075!#2$) edge [dotted, #3, #4] ($#1!0.15!#2$);
		\draw #2 edge [#3,#4] ($#2!0.075!#1$) edge [dotted, #3, #4] ($#2!0.15!#1$);
	}
}

\def\drawTopMatchingEdges(#1,#2)[#3][#4]{
	\gregLine((#1 1), (#2 1))[#3][#4][false];
	\gregLine((#1 2), (#2 2))[#3][#4][false];
	\gregLine((#1 3), (#2 3))[#3][#4][false];
	\gregLine((#1 4), (#2 4))[#3][#4][false];
	\gregLine((#1 5), (#2 5))[#3][#4][false];
}

\def\drawBottomMatchingEdges(#1,#2)[#3][#4][#5]{
	\gregLine((#1 6),(#2 6))[#3][#4][true];
	\ifthenelse{#5 > 1}{
		\gregLine((#1 7),(#2 7))[#3][#4][true];
	}{}
	\ifthenelse{#5 > 2}{
		\gregLine((#1 8),(#2 8))[#3][#4][true];
	}{}
}

\def\mark(#1)[#2] { 
	\ifthenelse{#2 = 0}{
		\node[markNode, draw=markColor0, fill=markColor0] ($#1$ m) at (#1) {};}{}
	\ifthenelse{#2 = 1}{
		\node[markNode, draw=markColor1, fill=markColor1] ($#1$ m) at (#1) {};}{}
	\ifthenelse{#2 = 2}{
		\node[markNode, draw=markColor2, fill=markColor2] ($#1$ m) at (#1) {};}{}
	\ifthenelse{#2 = 3}{
		\node[markNode, draw=markColor3, fill=markColor3] ($#1$ m) at (#1) {};}{}
	\ifthenelse{#2 = 4}{
		\node[markNode, draw=markColor4, fill=markColor4] ($#1$ m) at (#1) {};}{}
	\ifthenelse{#2 = 5}{
		\node[markNode, draw=markColor5, fill=markColor5] ($#1$ m) at (#1) {};}{}
	\ifthenelse{#2 = 6}{
		\node[markNode, draw=markColor6, fill=markColor6] ($#1$ m) at (#1) {};}{}
	\ifthenelse{#2 = 7}{
		\node[markNode, draw=markColor7, fill=markColor7] ($#1$ m) at (#1) {};}{}
	\ifthenelse{#2 = 8}{
		\node[markNode, draw=markColor8, fill=markColor8] ($#1$ m) at (#1) {};}{}
}

\def\drawChunk(#1,#2)(#3)[#4][#5](#6)[#7][#8]{ 
	\pgfmathtruncatemacro{\flip}{1}
	\ifthenelse{\equal{#8}{true}}{
		\pgfmathtruncatemacro{\flip}{-1};}{}
	
	\pgfmathtruncatemacro{\x}{#1}
	\pgfmathtruncatemacro{\y}{#2}
	\node[gregNode] at (\x, \y)   (#3 c) {};		
	\node[gregNode] at (\x+0.5*\chunkScale-2*\chunkScale, \y+\flip*1.25*\chunkScale)   (#3 1) {};		
	\node[gregNode] at (\x+0.5*\chunkScale-1*\chunkScale, \y+ \flip*1.25*\chunkScale)   (#3 2) {};		
	\node[gregNode] at (\x+0.5*\chunkScale, \y+\flip*1.25*\chunkScale)   (#3 3) {};		
	\node[gregNode] at (\x+0.5*\chunkScale+\chunkScale, \y+\flip*1.25*\chunkScale)   (#3 4) {};		
	\node[gregNode] at (\x+0.5*\chunkScale+2*\chunkScale, \y+\flip*1.25*\chunkScale)   (#3 5) {};		

	\ifthenelse{#4 > 2}{
		\node[gregNode] at (\x+0.5*\chunkScale, \y-\flip*\chunkScale)   (#3 8) {};
	}{}		
	\ifthenelse{#4 > 1}{
		\node[gregNode] at (\x+0.5*\chunkScale+\chunkScale, \y-\flip*\chunkScale)   (#3 7) {};		
	}{}
	\ifthenelse{#4 > 0}{
		\node[gregNode] at (\x+0.5*\chunkScale+2*\chunkScale, \y-\flip*\chunkScale)   (#3 6) {};		
	}{}
	
	\draw[cliqueEdge] (#3 c) -- (#3 1);
	\draw[cliqueEdge] (#3 c) -- (#3 2);
	\draw[cliqueEdge] (#3 c) -- (#3 3);
	\draw[cliqueEdge] (#3 c) -- (#3 4);
	\draw[cliqueEdge] (#3 c) -- (#3 5);
	\ifthenelse{#4 > 0}{
		\draw[cliqueEdge] (#3 c) -- (#3 6);	
	}{}	
	\ifthenelse{#4 > 1}{
		\draw[cliqueEdge] (#3 c) -- (#3 7);
	}{}
	\ifthenelse{#4 > 2}{
		\draw[cliqueEdge] (#3 c) -- (#3 8);
	}{}
	
	\ifthenelse{\equal{#7}{true}}{
		\node[#5] at ({#1}, {#2})   (k #3) {$#6$};}{}
}

\def\drawBigFig{
	\drawChunk(0,2*\wholeScale)(C)[3][above=20*\chunkScale](C_0)[true][false]

	\drawChunk(0,-\wholeScale)(C1)[3][below=20*\chunkScale](C_1)[true][true]
	\drawChunk(3*\wholeScale, -\wholeScale)(C2)[2][below=20*\chunkScale](C_2)[true][true]
	\drawChunk(6*\wholeScale, -\wholeScale)(C3)[1][below=20*\chunkScale](C_3)[true][true]
	\drawChunk(10*\wholeScale, -\wholeScale)(C4)[0][below=20*\chunkScale](C_4)[true][true]
	
	\drawTopMatchingEdges(C,C1)[matchEdge][matchEdgeColorA]
	\drawTopMatchingEdges(C,C2)[matchEdge][matchEdgeColorB]
	\drawTopMatchingEdges(C,C3)[matchEdge][matchEdgeColorC]
	\drawTopMatchingEdges(C,C4)[matchEdge][matchEdgeColorD]
	
	\drawBottomMatchingEdges(C,C1)[matchEdge][matchEdgeColorA][3]
	\drawBottomMatchingEdges(C,C2)[matchEdge][matchEdgeColorB][2]
	\drawBottomMatchingEdges(C,C3)[matchEdge][matchEdgeColorC][1]
	
	\draw[matchEdge,matchEdgeColorA] (C1 c) -- (C c) node[draw=none,fill=none, right, midway] {1};
	\draw[matchEdge,matchEdgeColorB] (C2 c) -- (C 8) node[draw=none,fill=none, above right, midway] {3};
	\draw[matchEdge,matchEdgeColorC] (C3 c) -- (C 7) node[draw=none,fill=none, above right, midway] {5};
	\draw[matchEdge,matchEdgeColorD] (C4 c) -- (C 6) node[draw=none,fill=none, above right, midway] {7};
	
	\mark(C c)[2]
	\mark(C 1)[0]
	\mark(C 2)[0]
	\mark(C 3)[0]
	\mark(C 4)[0]
	\mark(C 5)[0]
	\mark(C 6)[8]
	\mark(C 7)[6]
	\mark(C 8)[4]
	
	\mark(C1 c)[1]
	\mark(C1 1)[1]
	\mark(C1 2)[1]
	\mark(C1 3)[1]
	\mark(C1 4)[1]
	\mark(C1 5)[1]
	\mark(C1 6)[1]
	\mark(C1 7)[1]
	\mark(C1 8)[1]
	
	\mark(C2 c)[1]
	\mark(C2 1)[3]
	\mark(C2 2)[2]
	\mark(C2 3)[2]
	\mark(C2 4)[2]
	\mark(C2 5)[2]
	\mark(C2 6)[2]
	\mark(C2 7)[2]
	
	\mark(C3 c)[1]
	\mark(C3 1)[5]
	\mark(C3 2)[2]
	\mark(C3 3)[2]
	\mark(C3 4)[2]
	\mark(C3 5)[2]
	\mark(C3 6)[2]
	
	\mark(C4 c)[1]
	\mark(C4 1)[7]
	\mark(C4 2)[2]
	\mark(C4 3)[2]
	\mark(C4 4)[2]
	\mark(C4 5)[2]
}

\def\drawBigFigKey{
	\pgfmathtruncatemacro{\keyx}{0.9cm};
	\pgfmathtruncatemacro{\keyy}{5.5cm};
	
	 \foreach \color / \idx in 
	{markColor1/1, markColor2/2, markColor3/3, markColor4/4, markColor5/5, markColor6/6, markColor7/7, markColor8/8} 
	{ 
		\node [markNode, label={[xshift=0.8cm, yshift=-0.33cm]Phase \idx}, draw=\color, fill=\color] (key \idx) at (\keyx, \keyy - \idx *25pt){};
	}

	\draw[draw=black, line width=0.2mm] (24.3cm, -2.4cm) rectangle ++(4.2cm,7.9 cm);
}